\documentclass[runningheads]{llncs}

\usepackage[T1]{fontenc}
\usepackage{amsmath,amsfonts}
\usepackage{amssymb}
\usepackage{graphicx}
\usepackage{hhline}
\usepackage{multirow}
\usepackage[linesnumbered,lined,boxed,commentsnumbered]{algorithm2e}
\usepackage[hidelinks]{hyperref}
\usepackage{xcolor}
\usepackage[noabbrev,capitalize]{cleveref}

\usepackage[mathlines]{lineno}
\usepackage{paralist}

\newcommand{\Problem}[1]{\textit{#1}}

\newcommand{\dProb}{\Problem{DQMax\#SAT}}
\newcommand{\dPmd}{\Problem{Max\#SAT}}
\newcommand{\dPq}{\Problem{QBF}}
\newcommand{\dPdq}{\Problem{DQBF}}
\newcommand{\dPssat}{\Problem{SSAT}}
\newcommand{\dPds}{\Problem{DSSAT}}

\newcommand{\Boolean}{\mathbb{B}}

\newcommand{\isdef}{\stackrel{def}{=}}

\newcommand{\BExp}[1]{{\cal F}\langle #1\rangle}
\newcommand{\MExp}[1]{{\cal M}\langle #1\rangle}
\newcommand{\counting}{\reflectbox{{\sf R}}}
\newcommand{\size}[1]{\left|{#1}\right|}
\newcommand{\cardinal}[2]{|{#1}|_{#2}}
\newcommand{\pars}[1]{\left(#1\right)}

\newcommand{\xvar}[2]{x'_{#1,#2}}
\newcommand{\Max}[2]{\max\nolimits^{#1}{#2}.~}

\newcommand{\Prob}[1]{\ensuremath{\mathbb{P}\left[#1\right]}}

\newcommand{\todo}[1]{\textcolor{blue}{[todo: #1]}}
\newcommand{\missingref}{\textcolor{brown}{[MISSING REF]}}
\newcommand{\retq}[1]{\textcolor{magenta}{[ret. ques. : #1]}}

\newcommand{\comm}[1]{}

\graphicspath{{./figures}}

\begin{document}

\raggedbottom{}

\title{Function synthesis for maximizing model counting}

\author{Thomas~Vigouroux\orcidID{0000-0001-6396-0285} \and Marius~Bozga\orcidID{0000-0003-4412-5684} \and 
  Cristian~Ene\orcidID{0000-0001-6322-0383} \and 
  Laurent~Mounier\orcidID{0000-0001-9925-098X}}
\authorrunning{T.~Vigouroux et al.}

\institute{Univ. Grenoble Alpes, CNRS,
  Grenoble INP\footnote{Institute of Engineering Univ. Grenoble Alpes},
  VERIMAG, 38000 Grenoble, France
  \email{\{name.surname\}@univ-grenoble-alpes.fr}
  \url{https://www-verimag.imag.fr/}
}

\maketitle


\begin{abstract}
  Given a boolean formula $\phi(X, Y, Z)$, the \dPmd{}
  problem~\cite{fremont2017maxcount,vigouroux2022baxmc} asks for finding a partial model on the set of variables $X$, maximizing its number of projected models over
  the set of variables $Y$. We investigate a strict generalization of \dPmd{} allowing
  dependencies for variables in $X$, effectively turning it into a synthesis problem.  We show
  that this new problem, called \dProb{}, subsumes both the \dPdq{}~\cite{peterson1979mpalt} and
  \dPds{}~\cite{lee2021dssat} problems.  We
  provide a general resolution method, based on a reduction to \dPmd{}, together with two
  improvements for dealing with its inherent complexity.  We further discuss a concrete application
  of \dProb{} for symbolic synthesis of adaptive attackers in the field of program security.
  Finally, we report preliminary results obtained on the resolution of benchmark problems using a
  prototype \dProb{} solver implementation.

  \keywords{Function synthesis \and Model counting \and   \dPmd{} \and \dPdq{} \and \dPds{}  \and Adaptive attackers.}
\end{abstract}

\section{Introduction}\label{sec:introduction}

A major concern in software security are active adversaries, i.e., adversaries that can \textit{interact} with a target program 
by feeding inputs. Moreover, these adversaries can often make observations about the program execution through side-channels and/or 
legal outputs. In this paper, we consider \textit{adaptive} adversaries, i.e., adversaries that choose their inputs by taking advantage 
of previous observations.

In order to get an upper bound of the insecurity of a given program with respect to this class of adversaries, a possible approach is to synthesize 
the \textit{best} adaptive attack strategy. 
This can be modelled as finding a function $A$ (corresponding to the adversarial strategy) satisfying some logical formula $\Phi$ (capturing some 
combination of attack objectives). Actually, this corresponds to a classical functional synthesis problem.

Informally, in our case, given a Boolean relation $\Phi$ between output variables (observables) and input variables (attacker provided), 
our goal is to synthesize each input variable as a function on preceeding outputs satisfying $\Phi$.
In the literature, this synthesis problem is captured by the so-called Quantified Boolean Formulae
({\it QBF}) satisfiability problem~\cite{garey1979npcomplete,garey1975graphcoloring} and its generalization, the Dependency Quantified Boolean Formulae
(\dPdq{}) satisfiability problem~\cite{peterson1979mpalt}.

These existing qualitative frameworks are not sufficient in a security context: 
we are not only interested by adversaries able to succeed \textit{in all cases}, but rather for adversaries succeeding with ``a good probability''.
The Stochastic SAT (\dPssat{}) problem~\cite{papadimitriou1985ssat} was therefore proposed and replaces the classical
universal (resp. existential) quantifiers by \textit{counting} (resp. \textit{maximizing}) quantifiers.
This corresponds to finding the optimal inputs, depending on preceeding outputs, that maximize the number of models of $\Phi$,
hence the succeeding probability of the attack.  More recently, the Dependency Stochastic SAT (\dPds) problem~\cite{lee2021dssat} has been proposed as 
a strict generalization of the \dPssat{} problem by allowing explicit dependencies for maximizing variables,
in a similar way the \dPdq{} problem generalizes the \dPq{} problem.

Nonetheless, an additional complication is hindering the use of
quantitative stochastic frameworks in our security context.  In general,
the output variables in a program may hold expressions computed from one or
more secret variables.  Consequently, they rarely translate as
counting variables in a stochastic formula.  Most likely, the above-mentioned secret
variables translate into counting variables whereas the
observable variables need to be projected out when counting the models.
Yet, the output variables are mandatory to express the knowledge available
and the dependencies for synthesizing the attacker's optimal inputs.


As an example, we are interested in solving counting problems of the form:
\begin{multline*}
  \Max{\{z_1\}}{x_1} \Max{\{z_2\}}{x_2} \counting y_1.~ \counting y_2.~ \exists z_1.~ \exists z_2.~ \\
  (x_1 \Rightarrow y_2) \land (y_1 \Rightarrow x_2) \land (y_1 \lor z_2 \Leftrightarrow y_2 \land z_1)
\end{multline*}
which involve three distinct types of quantified variables and which are interpreted as follows:
synthesize for $x_1$ (respectively $x_2$) a boolean expression
$e_1$ (respectively  $e_2$), depending only on $z_1$ (respectively $z_2$), such that the formula
obtained after replacing $x_i$ by $e_i$ has a maximal number of  models \textit{projected} on the counting
variables $y_1, y_2$.

Notice that this problem generalizes in a non-trivial way three well-known existing problems:
\begin{inparaenum}[(i)]
  \item it generalizes the \dPmd{} problem~\cite{fremont2017maxcount,vigouroux2022baxmc} by allowing
        the maximizing variables to depend \emph{symbolically} on other variables;
  \item it lifts  the  \dPdq{} problem~\cite{peterson1979mpalt} to a quantitative problem, we do not
        want to check if there exist expressions $e_i$ working for all $y_1, y_2$, but to find
        expressions $e_i$ maximizing the number of models on $y_1, y_2$;
  \item it extends the \dPds{} problem~\cite{lee2021dssat} with the additional category of \emph{existential}
  variables, which can occur in the dependencies of maximizing variables, but which are projected for model counting.
\end{inparaenum}

Our contributions are the following:
\begin{itemize}
  \item We introduce formally  the \dProb{} problem as a new problem that arises naturally in the
        field of software security, and we show that it subsumes the \dPmd{}, \dPdq{} and \dPds{} problems.
  \item We develop a general  resolution method based on a reduction to \dPmd{} and further propose two improvements in order to deal with its inherent complexity:
    \begin{inparaenum}[(i)]
    \item an incremental method, that enables anytime resolution;
    \item a local method, allowing to split the initial problem into independent smaller
          sub-problems, enabling parallel resolution.
    \end{inparaenum}
  \item We provide two applications of  \dProb{} to software security:  we show that
        \emph{quantitative robustness} \cite{bardin2022quantreach} and  \emph{programs as
        information leakage-channels} \cite{smith2009foundations,phan2017sidechanatk}  can be
        systematically cast as instances of the \dProb{} problem.
  \item We provide a first working prototype solver for the \dProb{} problem and we apply it to the
        examples considered in this paper.
\end{itemize}





The paper is organized as follows. \cref{sec:problem} introduces formally the
\dProb{} problem and its relation with the \dPmd{}, \dPdq{} and \dPds{} problems.
\cref{sec:global,sec:incremental,sec:local} present the three different approaches we propose for
solving \dProb{}.  \cref{sec:appsec} shows concrete applications of \dProb{} in software
security, that is, for the synthesis of adaptive attackers.  Finally, \cref{sec:bench} provides
preliminary experimental results obtained with our prototype \dProb{} solver.
\cref{sec:relwork} discusses some references to related work and \cref{sec:conclusions} concludes
and proposes some extensions to address in the future.

\section{Problem statement}\label{sec:problem}

\subsection{Preliminaries}

Given a set $V$ of Boolean variables, we denote by $\BExp{V}$ (resp. $\MExp{V}$) the set of
Boolean formulae (resp. complete monomials) over $V$.  A model of a boolean formula $\phi \in
\BExp{V}$ is an assignement $\alpha_V : V \rightarrow \Boolean$ of
variables to Boolean values such that $\phi$ evaluates to $\top$ (that is,
$\mathit{true}$) on $\alpha_V$, it is denoted by $\alpha_V \models \phi$.  A formula is satisfiable if it has
at least one model $\alpha_V$.  A formula is valid (i.e., tautology) if any
assignement $\alpha_V$ is a model.  

Given a formula $\phi \in \BExp{V}$ we denote by $\cardinal{\phi}{V}$ the
number of its models, formally $\cardinal{\phi}{V} \isdef \size{\{\alpha_V
  : V \rightarrow \Boolean ~|~ \alpha_V \models \phi \}}$.  For a
partitioning $V = V_1 \uplus V_2$ we denote by $\cardinal{\exists
  V_2.~\phi}{V_1}$ the number of its $V_1$-projected models, formally
$\cardinal{ \exists V_2.~\phi }{V_1} \isdef \size{\{ \alpha_{V_1} : V_1
  \rightarrow \Boolean ~|~ \exists \alpha_{V_2} : V_2\rightarrow \Boolean.~
  \alpha_{V_1} \uplus \alpha_{V_2} \models \phi \}}$.  Note that
in general $\cardinal{\exists V_2.~\phi}{V_1} \le \cardinal{\phi}{V}$ with
equality only in some restricted situations (e.g. when $V_1$ is an independent support of the
formula~\cite{ChakrabortyMV14independent}).

Let $V$, $V'$, $V''$ be arbitrary sets of Boolean variables.  Given a
Boolean formula $\phi \in \BExp{V}$ and a substitution $\sigma : V'
\rightarrow \BExp{V''}$ we denote by $\phi[\sigma]$ the Boolean formula in
$\BExp{(V \setminus V') \cup V''}$ obtained by replacing in $\phi$ all
occurrences of variables $v'$ from $V'$ by the associated formula
$\sigma(v')$.

\subsection{Problem Formulation}

\begin{definition}[\dProb{} problem] \label{def:problem}
Let $X=\{x_1,...,x_n\}$, $Y$, $Z$ be pairwise disjoint finite sets
of Boolean variables, called respectively \emph{maximizing},
\emph{counting} and \emph{existential} variables.  The \dProb{} problem is
specified as:
\begin{equation} \label{eq:problem}
  \max\nolimits^{H_1}x_1.~ ...~ \max\nolimits^{H_n} x_n.~
  \counting Y.~ \exists Z.~ \Phi(X,Y,Z)
\end{equation}
where $H_1, ..., H_n \subseteq Y \cup Z$ and $\Phi \in \BExp{X \cup Y \cup
  Z}$ are respectively the \emph{dependencies} of maximizing variables and
the \emph{objective} formula.
\end{definition}
The solution to the problem is a substitution $\sigma^*_X : X \rightarrow \BExp{Y \cup Z}$ associating
formulae on counting and existential variables to maximizing variables such
that
\begin{inparaenum}[(i)]
  \item $\sigma^*_X(x_i) \in \BExp{H_i}$, for all $i \in [1,n]$ and
  \item $\cardinal{\exists Z.~ \Phi[\sigma^*_X]}{Y}$ is maximal.
\end{inparaenum}
That means, the chosen substitution conforms to dependencies on maximizing
variables and guarantees the objective holds for the largest number of
models projected on the counting variables.

\begin{example}\label{ex:running1}
  Consider the problem:
  \[
    \max\nolimits^{\{z_1, z_2\}}x_1.~ \counting y_1.~ \counting y_2.~ \exists z_1.~ \exists z_2.~
    (x_1 \Leftrightarrow y_1) \land (z_1 \Leftrightarrow y_1 \lor y_2) \land (z_2 \Leftrightarrow y_1 \land y_2)
  \]
  Let $\Phi$ denote the objective formula. In this case, $\BExp{\{z_1, z_2\}} = \{\top,
  \bot, z_1, \overline{z_1}, z_2, \overline{z_2}, z_1 \lor z_2, \overline{z_1} \lor z_2, z_1 \lor
  \overline{z_2}, \overline{z_1} \lor \overline{z_2}, z_1 \land z_2, \overline{z_1} \land z_2, z_1
  \land \overline{z_2}, \overline{z_1} \land \overline{z_2}, z_1 \Leftrightarrow z_2, \overline{z_1 \Leftrightarrow z_2} \}$, and one shall consider every
  possible substitution.  One can compute for instance 
  $\Phi[x_1 \mapsto \overline{z_1} \land \overline{z_2}] \equiv((\overline{z_1} \land
  \overline{z_2}) \Leftrightarrow y_1) \land (z_1 \Leftrightarrow y_1 \lor y_2) \land (z_2 \Leftrightarrow y_1 \land y_2)$
  which only has one model ($\{y_1 \mapsto \bot, y_2 \mapsto \top, z_1 \mapsto \top, z_2 \mapsto
  \bot\}$) and henceforth $\cardinal{\exists z_1.~ \exists z_2.~
  \Phi[x_1 \mapsto \overline{z_1} \land \overline{z_2}]}{\{y_1, y_2\}} = 1$.
  Overall, for this problem there exists four possible maximizing substitutions
  $\sigma^*$ respectively $x_1 \mapsto z_1$, $x_1 \mapsto z_2$, $x_1 \mapsto z_1 \lor z_2$, $x_1 \mapsto z_1 \land z_2$ such that for all of them $\cardinal{\exists z_1.~\exists z_2.~\Phi[\sigma^*]}{\{y_1, y_2\}} = 3$.
\end{example}

\begin{example} \label{ex:running2}
  Let us consider the following problem:
  \begin{multline*}
    \Max{\{z_1\}}{x_1} \Max{\{z_2\}}{x_2} \counting y_1.~ \counting y_2.~ \exists z_1.~ \exists z_2.~ \\
    (x_1 \Rightarrow y_2) \land (y_1 \Rightarrow x_2) \land (y_1 \lor z_2 \Leftrightarrow y_2 \land z_1)
  \end{multline*}
   \noindent Let $\Phi$ denote the associated objective formula.
   An optimal solution is $x_1 \mapsto \bot, x_2 \mapsto \overline{z_2}$ and one can check that
   $\cardinal{\exists z_1.~ \exists z_2.~ \Phi[x_1 \mapsto \bot, x_2 \mapsto \overline{z_2}]}{\{y_1,
   y_2\}} = 3$. Moreover,  on can notice that there do not exist expressions $e_1 \in
   \BExp{\{z_1\}}$ (respectively $e_2 \in  \BExp{\{z_2\}}$),  such that $\exists z_1.~ \exists z_2.~
   \Phi[x_1 \mapsto e_1, x_2 \mapsto e_2]$ admits the model $y_1 \mapsto \top, y_2 \mapsto \bot$.
\end{example}

The following proposition provides an upper bound on the number of models
corresponding to the solution of (\ref{eq:problem}) computable using
projected model counting.

\begin{proposition}\label{prop:upper-bound}
  For any substitution $\sigma_X : X \rightarrow \BExp{Y \cup Z}$ it holds
  \[ \cardinal{\exists Z.~ \Phi[\sigma_X]}{Y} \le \cardinal{ \exists X.~\exists Z.~\Phi}{Y}.\]
\end{proposition}

\subsection{Hardness of \dProb{}}

We briefly discuss now the relationship between the \dProb{} problem and the \dPmd{},
\dPdq{} and \dPds{} problems.  It turns out that \dProb{} is at least as hard as all of them, as illustrated by the
following reductions.

\subsubsection{\dProb{} is at least as hard as \dPmd{}:}
Let $X = \{x_1,...,x_n\}$, $Y$, $Z$ be pairwise disjoint finite
sets of Boolean variables, called \emph{maximizing},
\emph{counting} and \emph{existential} variables.  The \dPmd{} problem
\cite{fremont2017maxcount} specified as
\begin{equation} \label{md:problem}
  \max x_1.~ \ldots \max x_n.~ \counting Y.~  \exists Z.~ \Phi(X,Y,Z)
\end{equation}
asks for finding an assignement $\alpha^*_X : X \rightarrow \Boolean$
of maximizing variables to Boolean values such that $\cardinal{\exists Z.~
  \Phi[\alpha^*_X]}{Y}$ is maximal.  It is immediate to see that the \dPmd{}
problem is the particular case of the \dProb{} problem where there are no
dependencies, that is, $H_1 = H_2 = ... = H_n = \emptyset$.

\subsubsection{\dProb{} is at least as hard as \dPdq{}:}
Let $X= \{x_1,...,x_n\}$, $Y$ be disjoint finite sets of
Boolean variables and let $H_1, ..., H_n \subseteq Y$.  The \dPdq{} problem
\cite{peterson1979mpalt} asks, given a \dPdq{} formula:
\begin{equation} \label{dq:problem}
\forall Y.~ \exists^{H_1} x_1.~ ...~ \exists^{H_n} x_n.~  \Phi(X,Y)
\end{equation}
to synthesize a substitution $\sigma^*_X : X \rightarrow \BExp{Y}$ whenever
one exists such that
\begin{inparaenum}[(i)]
\item $\sigma^*_X(x_i) \in \BExp{H_i}$, for all $i \in [1,n]$ and
\item $\Phi[\sigma^*_X]$ is valid.
\end{inparaenum}
The \dPdq{} problem is reduced to the \dProb{} problem:
\begin{equation} \label{eqdq:problem}
  \max\nolimits^{H_1}x_1.~ ...~ \max\nolimits^{H_n} x_n.~\counting Y.~ \Phi(X,Y)
\end{equation}
By solving \eqref{eqdq:problem} one can solve the initial \dPdq{} problem
\eqref{dq:problem}.  Indeed, let $\sigma^*_X : X \rightarrow \BExp{Y}$ be a
solution for \eqref{eqdq:problem}.  Then, the \dPdq{} problem admits a
solution if and only if $\cardinal{\Phi[\sigma^*_X]}{Y} = 2^{|Y|}$.
Moreover, $\sigma^*_X$ is a solution for the problem
\eqref{dq:problem} because
\begin{inparaenum}[(i)]
\item $\sigma^*_X$ satisfies dependencies and
\item $\Phi[\sigma^*_X]$ is valid as it belongs to $\BExp{Y}$ 
  and has $2^{|Y|}$ models.
\end{inparaenum}
Note that through this reduction of \dPdq{} to \dProb{}, the maximizing quantifiers in \dProb{} can
be viewed as Henkin quantifiers~\cite{henkin1965quantifier} in \dPdq{} with a quantitative flavor.

\subsubsection{\dProb{} is at least as hard as \dPds{}:} \label{sec:hard:dssat}
Let $X = \{x_1, ..., x_n\}$, $Y = \{y_1, ..., y_m\}$ be disjoint finite sets of
variables.  A \dPds{} formula is of the form:
\begin{equation} \label{eq:dssatdef}
  \max\nolimits^{H_1}x_1.~ ...~ \max\nolimits^{H_n} x_n.
  ~\counting^{p_1} y_1.~ ...~ \counting^{p_m} y_m.~ \Phi(X,Y)
\end{equation}
where $p_1, ..., p_m \in [0, 1]$ are respectively the probabilities of
variables $y_1,..., y_m$ to be assigned $\top$ and $H_1, ...,H_n \subseteq
Y$ are respectively the dependency sets of variables $x_1,..., x_n$.
Given a \dPds{} formula (\ref{eq:dssatdef}), the probability of an
assignement $\alpha_Y:Y \rightarrow \Boolean$ is defined as
\[
  \Prob{\alpha_Y} \isdef \prod_{i=1}^m \begin{cases}
    p_i & \text{if } \alpha_Y(y_i) = \top \\
    1 - p_i & \text{if } \alpha_Y(y_i) = \bot \\
  \end{cases}
\]
This definition is lifted to formula $\Psi \in \BExp{Y}$ by summing up the
probabilities of its models, that is, $\Prob{\Psi} \isdef \sum_{\alpha_Y \models
  \Psi} \Prob{\alpha_Y}$.

The \dPds{} problem~\cite{lee2021dssat} asks, for a given formula (\ref{eq:dssatdef}),
to synthesize a substitution $\sigma^*_X : X \rightarrow \BExp{Y}$ such that
\begin{inparaenum}[(i)]
\item $\sigma^*_X(x_i) \in \BExp{H_i}$, for all $i \in [1,n]$ and
\item $\Prob{\Phi[\sigma^*_X]}$ is maximal.
\end{inparaenum}
If $p_1 = ... = p_m = \frac{1}{2}$ then for any substitution $\sigma_X : X
\rightarrow \BExp{Y}$ it holds $\Prob{\Phi[\sigma_X]} =
\frac{\cardinal{\Phi[\sigma_X]}{Y}}{2^m}$.  In this case, it is
immediate to see that solving (\ref{eq:dssatdef}) as a \dProb{} problem
(i.e., by ignoring probabilities) would solve the original \dPds{} problem.
Otherwise, in the general case, one can use existing techniques such as
\cite{chakraborty2015unweight} to transform arbitrary \dPds{} problems
(\ref{eq:dssatdef}) into equivalent ones where all probabilities are
$\frac{1}{2}$ and solve them as above.

Note that while the reduction above from \dPds{} to \dProb{} seems to
indicate the two problems are rather similar, a reverse reduction from
\dProb{} to \dPds{} seems not possible in general. That is, recall that
\dProb{} allows for a third category of \emph{existential} variables $Z$
which can occur in the dependencies sets $H_i$ and which are not used for
counting but are projected out. Yet, such problems arise naturally in our
application domain as illustrated later in
section~\ref{sec:appsec}.  If no such existential variables exists or if
they do not occur in the dependencies sets then one can apriori project
them from the objective $\Phi$ and syntactically reduce \dProb{} to \dPds{}
(i.e., adding $\frac{1}{2}$ probabilities on counting variables).  However,
projecting existential variables in a brute-force way may lead to an
exponential blow-up of the objective formula $\Phi$, an issue already
explaining the hardness of projected model counting vs model counting
\cite{aziz2015projmc,lagniez2019projmc}.  Otherwise, in case of dependencies on
existential variables, it is an open question if any direct reduction exists
as these variables do not fit into the two categories of variables
(counting, maximizing) occurring in \dPds{} formula.

\section{Global method}\label{sec:global}
We show in this section that the \dProb{} problem can be directly reduced to
a \dPmd{} problem with an exponentially larger number of maximizing
variables and exponentially bigger objective formula.

First, recall that any boolean formula $\varphi \in \BExp{H}$ can be
written as a finite disjunction of a subset $M_\varphi$ of complete
monomials from $\MExp{H}$, that is, such that the following
equivalences hold:
\[
  \varphi ~\Longleftrightarrow~ \vee_{m \in M_\varphi} m
~\Longleftrightarrow~ \vee_{m \in \MExp{H}} ([\![m \in M_\varphi]\!]
\wedge m)
\]
Therefore, any formula $\varphi \in \BExp{H}$ is uniquely
\emph{encoded} by the set of boolean values $[\![m \in M_\varphi]\!]$
denoting the membership of each complete monomial $m$ to $M_\varphi$.
We use this idea to encode the substitution of a maximizing variable
$x_i$ by some formula $\varphi_i \in \BExp{H_i}$ by using a set of
boolean variables $(\xvar{i}{m})_{m\in \MExp{H_i}}$ denoting
respectively $[\![m \in M_{\varphi_i}]\!]$ for all
$m \in \MExp{H_i}$.
We now define the following \dPmd{} problem:
\begin{multline} \label{eq:md:problem}
    (\max \xvar{1}{m}.)_{m \in \MExp{H_1}}~ ... ~(\max \xvar{n}{m}.)_{m \in \MExp{H_n}} ~ \counting Y.~ \exists Z.~ \exists X.~ \\
    \Phi(X,Y,Z) \wedge \bigwedge_{i\in[1,n]} \left(x_i \Leftrightarrow \vee_{m \in \MExp{H_i}} (\xvar{i}{m} \wedge m)\right)
\end{multline}
The next theorem establishes the relation between the two problems.
\begin{theorem} \label{thm:reduce}
$\sigma^*_X = \{x_i \mapsto \varphi^*_i\}_{i \in [1,n]}$ is a solution to the
problem \dProb{} (\ref{eq:problem}) if and only if $\alpha^*_{X'} = \{ \xvar{i}{m}
\mapsto [\![ m \in M_{\varphi_i^*} ]\!] \}_{i\in [1,n], m\in \MExp{H_i}}$ is a solution
to \dPmd{} problem \eqref{eq:md:problem}.
\end{theorem}
\begin{proof}
  Let us denote
  \[
    \Phi'(X', X, Y, Z) \isdef \Phi(X,Y,Z) \wedge \bigwedge_{i\in[1,n]} \left(x_i \Leftrightarrow
  \vee_{m \in \MExp{H_i}} (\xvar{i}{m} \wedge m)\right)
\]
Actually, for any $\Phi \in \BExp{X \cup Y \cup Z}$
for any $\varphi_1 \in \BExp{H_1}, ..., \varphi_n \in \BExp{H_n}$ the
following equivalence is valid:
\begin{multline*}
\Phi(X,Y,Z)[\{x_i \mapsto \varphi_i\}_{i\in [1,n]}] \Leftrightarrow \\ \left(\exists X.~
  \Phi'(X', X,Y,Z)\right) \left[ \{\xvar{i}{m} \mapsto [\![ m \in M_{\varphi_i}]\!]\}_{i \in [1, n],  m \in \MExp{H_i}}\right] 
\end{multline*}
Consequently, finding the substitution $\sigma_X$ which maximize the
number of $Y$-models of the left-hand side formula (that is, of $\exists
Z.~\Phi(X,Y,Z)$) is actually the same as finding the valuation
$\alpha_{X'}$ which maximizes the number of $Y$-models of the right-hand
side formula (that is, $\exists Z.~\exists
X.~ \Phi'(X', X, Y,Z)$). \squareforqed \end{proof}

\begin{example}\label{ex:global}
  \cref{ex:running1} is reduced to the following:
  \begin{multline*}
  \max \xvar{1}{z_1z_2}.~\max \xvar{1}{z_1\overline{z_2}}.~\max \xvar{1}{\overline{z_1}z_2}.~\max \xvar{1}{\overline{z_1}\overline{z_2}}.~ \counting y_1.~ \counting y_2.~ \exists z_1.~ \exists z_2.~\exists x_1.~ \\
  \shoveleft{(x_1 \Leftrightarrow y_1) \land (z_1 \Leftrightarrow y_1 \lor y_2) \land (z_2 \Leftrightarrow y_1 \land y_2) \land} \\
    (x_1 \Leftrightarrow \pars{\pars{\xvar{1}{z_1z_2} \wedge z_1 \wedge z_2} \vee
    \pars{\xvar{1}{z_1 \overline{z_2}} \wedge z_1 \wedge \overline{z_2}} \vee
    \pars{\xvar{1}{\overline{z_1} z_2} \wedge \overline{z_1} \wedge z_2} \vee
    \pars{\xvar{1}{\overline{z_1}\overline{z_2}} \wedge \overline{z_1} \wedge \overline{z_2}}})
  \end{multline*}
  One possible answer is $\xvar{1}{z_1 z_2} \mapsto \top, \xvar{1}{z_1 \overline{z_2}} \mapsto
  \top, \xvar{1}{\overline{z_1}z_2} \mapsto \bot, \xvar{1}{\overline{z_1}\overline{z_2}} \mapsto \bot$. This yields
  the solution $\sigma_X(x_1) = \pars{z_1 \land z_2} \lor \pars{z_1 \land \overline{z_2}} = z_1$ which is one of the
  optimal solutions as explained in \cref{ex:running1}.
\end{example}


\section{Incremental method}\label{sec:incremental}
In this section we propose a first improvement with respect to the reduction in the previous
section. It allows to control the blow-up of the objective formula in the reduced
\dPmd{} problem through an incremental process. Moreover, it allows in practice to find earlier good
approximate solutions.

The incremental method consists in solving a sequence of related \dPmd{}
problems, each one obtained from the original \dProb{} problem and a reduced
set of dependencies $H_1' \subseteq H_1$, \dots, $H_n' \subseteq H_n$.
Actually, if the sets of dependencies $H_1'$, \dots, $H_n'$ are chosen such
that to augment progressively from $\emptyset$, \dots, $\emptyset$ to $H_1$, \dots, $H_n$ by
increasing only one of $H_i'$ at every step then
\begin{inparaenum}[(i)]
\item it is possible to build every such \dPmd{} problem from the previous one
by a simple syntactic transformation and
\item most importantly, it is possible to steer the search for its solution knowing the solution of the previous one.
\end{inparaenum}

The incremental method relies therefore on an oracle procedure {\tt
max\#sat} for solving \dPmd{} problems.  We assume this procedure takes as
inputs the sets $X$, $Y$, $Z$ of maximizing, counting and existential
variables, an objective formula $\Phi \in \BExp{X \cup Y \cup Z}$, an
initial assignment $\alpha_0 : X \rightarrow \Boolean$ and a filter formula
$\Psi \in \BExp{X}$.  The last two parameters are essentially used to
restrict the search for maximizing solutions and must satisfy:
\begin{itemize}
\item $\Psi[\alpha_0] = \top$, that is, the initial assignment $\alpha_0$ is a model of $\Psi$ and 
\item forall $\alpha : X \rightarrow \Boolean$ if $\alpha \nvDash \Psi$ then $\cardinal{\exists
      Z.~\Phi[\alpha]}{Y} \le \cardinal{\exists Z.~\Phi[\alpha_0]}{Y}$, that is, any assignment
      $\alpha$ outside the filter $\Psi$ is at most as good as the assignement $\alpha_0$.
\end{itemize}
Actually, whenever the conditions hold, the oracle can safely restrict the search for the
optimal assignements within the models of $\Psi$.  The oracle produces as
output the optimal assignement $\alpha^* : X \rightarrow \Boolean$ solving
the \dPmd{} problem.

The incremental algorithm proposed in \cref{alg:incremental} proceeds as follows:
\begin{itemize}
  \item at lines 1-5 it prepares the arguments for the first call of the \dPmd{} oracle, that is,
        for solving the problem where $H_1' = H_2' = ... = H_n' = \emptyset$,
  \item at line 7 it calls to the \dPmd{} oracle, 
  \item at lines 9-10 it chooses an index $i_0$ of some dependency set $H_i' \not= H_i$ and a
    variable $u \in H_{i_0} \setminus H'_{i_0}$ to be considered in addition for the next step,
  \item at lines 11-19 it prepares the argument for the next call of the \dPmd{} oracle, that
    is, it updates the set of maximizing variables $X'$, it refines the objective formula
    $\Phi'$, it defines the new initial assignement $\alpha_0'$ and the new filter $\Psi'$ using
    the solution of the previous problem,
  \item at lines 6,20,22 it controls the main iteration, that is, keep going as long as sets
    $H_i'$ are different from $H_i$,
  \item at line 23 it builds the expected solution, that is, convert the Boolean solution
    $\alpha'^*$ of the final \dPmd{} problem where $H_i' = H_i$ for all $i\in [1,n]$ to the
    corresponding substitution $\sigma_X^*$.
\end{itemize}

\begin{algorithm}[t]
  \DontPrintSemicolon
  \SetKwInOut{Input}{input}
  \SetKwInOut{Output}{output}
  \SetKwFunction{Choose}{choose}
  \SetKwFunction{MaxDSat}{max\#sat}
  \SetKw{KwReturn}{return}
  \Input{$X = \{x_1,...,x_n\}$, $Y$, $Z$, $H_1$, ..., $H_n$, $\Phi$ }
  \Output{$\sigma_X^*$}
  \BlankLine
  $H_i' \leftarrow \emptyset$ for all $i\in[1,n]$ \;
  $X' \leftarrow \{ \xvar{i}{\top} \}_{i\in[1,n]}$ \;
  $\Phi' \leftarrow \Phi \wedge \bigwedge\nolimits_{i\in[1,n]} (x_i \Leftrightarrow \xvar{i}{\top}) $ \;
  $\alpha_0' \leftarrow \{\xvar{i}{\top} \mapsto \bot\}_{i\in[1,n]}$ \;
  $\Psi' \leftarrow \top$ \;
  \Repeat{$H_i' = H_i$ for all $i\in[1,n]$}{
  $\alpha'^* \leftarrow \MaxDSat(X', Y, Z \cup X, \Phi', \alpha_0', \Psi')$ \;
  \If{$H_i' \not= H_i$ for some $i\in[1,n]$}{
    $i_0 \leftarrow \Choose( \{ i\in[1,n] ~|~ H_i' \not= H_i \})$ \; \label{alg:inc:i}
     $u \leftarrow \Choose( H_{i_0} \setminus H'_{i_0})$ \; \label{alg:inc:u}
     $\alpha_0' \leftarrow \alpha'^*$ \;
     $\Psi' \leftarrow \bot$ \;
     \ForEach{$m \in \MExp{H'_{i_0}}$}{
       $X' \leftarrow (X' \setminus \{ \xvar{i_0}{m} \}) \cup \{\xvar{i_0}{mu},\xvar{i_0}{m\bar{u}}\}$ \;
       $\Phi' \leftarrow \Phi' [ \xvar{i_0}{m} \mapsto (\xvar{i_0}{mu} \wedge u) \vee (\xvar{i_0}{m\bar{u}} \wedge \bar{u})]$ \;
       $\alpha'_0 \leftarrow (\alpha_0' \setminus \{\xvar{i_0}{m} \mapsto \_ \}) \cup
         \{ \xvar{i_0}{mu}, \xvar{i0}{m\bar{u}} \mapsto \alpha_0'(\xvar{i_0}{m}) \}$ \;
       $\Psi' \leftarrow \Psi' \vee (\xvar{i_0}{mu} \not\Leftrightarrow \xvar{i_0}{m\bar{u}})$ \;
       }
     $\Psi' \leftarrow \Psi' \vee \bigwedge\nolimits_{x \in X'} (x \Leftrightarrow \alpha'_0(x))$ \;
     $H'_{i_0} \leftarrow H'_{i_0} \cup \{u\}$
  }  
  }
  $\sigma_X^* \leftarrow \{ x_i \mapsto \vee_{m \in \MExp{H_i}} (\alpha'^*(\xvar{i}{m}) \wedge m) \}_{i\in[1,n]}$ \label{alg:incremental:final} \;
  \caption{Incremental Algorithm\label{alg:incremental}}
\end{algorithm}

Finally, note that the application of substitution at line 15 can be done
such that to preserve the CNF form of $\Phi'$.
That is, the substitution proceeds clause by clause by using the following equivalences, for every formula $\psi$:
\begin{align*}
  (\psi \vee \xvar{i_0}{m})[ \xvar{i_0}{m} \mapsto (\xvar{i_0}{mu} \wedge u) \vee (\xvar{i_0}{m\bar{u}} \wedge \bar{u})] & \Leftrightarrow \\
  & \hspace*{-4cm}(\psi \vee \xvar{i_0}{mu} \vee \xvar{i_0}{m\bar{u}}) \wedge
  (\psi \vee \xvar{i_0}{mu} \vee \bar{u}) \wedge
  (\psi \vee \xvar{i_0}{m\bar{u}} \vee u) \\
  (\psi \vee \overline{\xvar{i_0}{m}})[ \xvar{i_0}{m} \mapsto (\xvar{i_0}{mu} \wedge u) \vee (\xvar{i_0}{m\bar{u}} \wedge \bar{u})] & \Leftrightarrow 
  (\psi \vee \overline{ \xvar{i_0}{mu}} \vee \bar{u})
  (\psi \vee \overline{ \xvar{i_0}{m\bar{u}}} \vee u)
\end{align*}
%
\begin{theorem}
Algorithm~\ref{alg:incremental} is correct for solving the \dProb{} problem (\ref{eq:problem}).
\end{theorem}
\begin{proof}
The algorithm terminates after $1+\sum_{i \in [1,n]} \size{H_i}$ oracle calls.
  Moreover, every oracle call solves correctly
the \dPmd{} problem corresponding to \dProb{} problem
\begin{equation} \nonumber
\max\nolimits^{H'_1}x_1.~ ...~ \max\nolimits^{H'_n} x_n.~  \counting Y.~ \exists Z.~ \Phi(X,Y,Z)
\end{equation}
This is an invariance property provable by induction.  It holds by
construction of $X'$, $\Phi'$, $\alpha_0'$, $\Psi'$ at the initial step.
Then, it is preserved from one oracle call to the next one i.e., $X'$ and
$\Phi'$ are changed such that to reflect the addition of the variable $u$
of the set $H'_{i_0}$.  The new initial assignement $\alpha_0'$ is obtained
\begin{inparaenum}[(i)]
\item by replicating the optimal value $\alpha'^{*}(\xvar{i_0}{m})$ to the newly
introduced $\xvar{i_0}{m u},\xvar{i_0}{m \bar{u}}$ variables derived from
$\xvar{i_0}{m}$ variable (line 16) and 
\item by keeping the optimal value $\alpha'^{*}(\xvar{i}{m})$ for other variables (line 11).
\end{inparaenum}
As such, for the new problem, the assignement $\alpha_0'$ has exactly the
same number of $Y$-projected models as the optimal assignement $\alpha'^*$
had on the previous problem.  The filter $\Psi'$ is built such that to
contain this new initial assignment $\alpha_0'$ (line 19) as well as any
other assignement that satisfies
$\xvar{i_0}{mu} \not\Leftrightarrow \xvar{i_0}{m\bar{u}}$ for some monomial
$m$ (lines 12, 17). This construction guarantees that, any assignment which
does not satisfy the filter $\Psi'$ reduces precisely to an assignment of
the previous problem, other than the optimal one $\alpha'^*$, and
henceforth at most as good as $\alpha'_0$ regarding the number of
$Y$-projected models. Therefore, it is a sound filter and can be used to
restrict the search for the new problem. The final oracle call corresponds
to solving the complete \dPmd{} problem~(\ref{eq:md:problem}) and it will
therefore allow to derive a correct solution to the initial \dProb{}
problem (\ref{eq:problem}). \squareforqed
\end{proof}

\begin{example}\label{ex:incremental}
  Let reconsider \cref{ex:running1}.  The
incremental algorithm will perform 3 calls to the \dPmd{} oracle.  The first
call corresponds to the problem
\begin{multline*}
\max \xvar{1}{\top}.~\counting y_1.~\counting y_2.~\exists z_1.~\exists z_2.~\exists x_1.~\\
(x_1 \Leftrightarrow y_1) \land (z_1 \Leftrightarrow y_1 \lor y_2) \land (z_2 \Leftrightarrow y_1 \land y_2) \land (x_1 \Leftrightarrow \xvar{1}{\top})
\end{multline*}
A solution found by the oracle is e.g., $\xvar{1}{\top} \mapsto \bot$ which has 2 projected models.
If $z_1$ is added to $H'_1$, the second call corresponds to the refined \dPmd{} problem:
\begin{multline*}
\max \xvar{1}{z_1}.~ \max \xvar{1}{\bar{z_1}}~\counting y_1.~\counting y_2.~\exists z_1.~\exists z_2.~\exists x_1.~\\
(x_1 \Leftrightarrow y_1) \land (z_1 \Leftrightarrow y_1 \lor y_2) \land (z_2 \Leftrightarrow y_1 \land y_2) \land
  (x_1 \Leftrightarrow \xvar{1}{z_1} \wedge z_1 \vee \xvar{1}{\bar{z_1}} \wedge \bar{z_1})
\end{multline*}
A solution found by the oracle is e.g., $\xvar{1}{z_1} \mapsto
\top, \xvar{1}{\bar{z_1}}\mapsto \bot$ which has 3 projected models.
Finally, $z_2$ is added to $H'_1$ therefore the third call corresponds to
the complete \dPmd{} problem as presented in Example~\ref{ex:global}.
The solution found by the oracle is the same as in Example~\ref{ex:global}.
\end{example}

A first benefit of \cref{alg:incremental} is the fact that it opens the
door to any-time approaches to solve the \dProb{} problem.  Indeed, the
distance between the current and the optimal solution (that is, the
relative ratio between the corresponding number of  $Y$-projected models) can be estimated using the
upper bound provided by Prop.~\ref{prop:upper-bound}.  Hence, one could
stop the search at any given iteration as soon as some threshold is
reached and construct the returned value $\sigma_X$ similarly as
in \cref{alg:incremental:final} of \cref{alg:incremental}.  In this case
the returned $\sigma_X$ would be defined as $\sigma_X = \{ x_i
\mapsto \vee_{m \in \MExp{H_i'}} (\alpha'^*(\xvar{i}{m}) \wedge m) \}_{i\in[1,n]}$ (note here that
the monomials are selected from $H_i'$ instead of $H_i$).  

Another benefit of the incremental approach is that it is applicable without any assumptions on the
underlying \dPmd{} solver. Indeed, one can use $\Psi'$ in \cref{alg:incremental}
by solving the \dPmd{} problem corresponding to $\Phi' \land \Psi'$, and return the found
solution. Even though the $\alpha_0'$ parameter requires an adaptation of the \dPmd{} solver in order to
ease the search of a solution, one could still benefit from the incremental resolution of
\dProb{}. Notice that a special handling of the $\Psi'$ parameter by the solver would
avoid complexifying the formula passed to the \dPmd{} solver and still steer the search properly.

\section{Local method}\label{sec:local}
The local resolution method allows to compute the solution of an initial
\dProb{} problem by combining the solutions of two strictly smaller and
independent \dProb{} sub-problems derived syntactically from the initial
one.  The local method applies only if either 1) some counting or existential
variable $u$ is occurring in all dependency set; or 2) if there is some maximizing variable having an empty dependency set.  That is, in contrast
to the global and incremental methods, the local method is applicable
only in specific situations.

Let us consider a \dProb{} problem of form (\ref{eq:problem}). 
Given a variable $v$, let
$\Phi_v \isdef \Phi[v\mapsto \top]$,
$\Phi_{\bar{v}} \isdef \Phi[v \mapsto \bot]$ be the two cofactors on
variable $v$ of the objective $\Phi$.

\subsection{Reducing common dependencies}\label{subsec:local1}
 Let us consider now a
variable $u$ which occurs in all dependency sets $H_i$ and let
us consider the following
$u$-reduced \dProb{} problems:
\begin{eqnarray}
\max\nolimits^{H_1\setminus\{u\}} x_1.~... \max\nolimits^{H_n \setminus\{u\}} x_n.~\counting\,Y\setminus\{u\}.~\exists\,Z\setminus\{u\}.~\Phi_u \label{eq:problem:u}\\
\max\nolimits^{H_1\setminus\{u\}} x_1.~... \max\nolimits^{H_n \setminus\{u\}} x_n.~\counting\,Y\setminus\{u\}.~\exists\,Z\setminus\{u\}.~\Phi_{\bar{u}} \label{eq:problem:baru}
\end{eqnarray}
Let $\sigma^*_{X,u}$, $\sigma^*_{X,{\bar{u}}}$ denote respectively the solutions to the problems above.

\begin{theorem}\label{thm:local}
  If either
  \begin{enumerate}[(i)]
    \item $u\in Y$ or
    \item $u \in Z$ and $u$ is functionally dependent on counting variables $Y$ within the objective
      $\Phi$ (that is, for any valuation $\alpha_Y : Y \to \Boolean$, at most one of
      $\Phi[\alpha_Y][u \mapsto \top]$ and $\Phi[\alpha_Y][u \mapsto \bot]$ is satisfiable).
  \end{enumerate}
  then $\sigma_X^*$ defined as
  $$\sigma_X^*(x_i) \isdef \pars{u \wedge \sigma^*_{X,u}(x_i)} \vee \pars{\bar{u} \wedge \sigma^*_{X,\bar{u}}(x_i)} \mbox{ for all } i\in [1,n]$$
  is a solution to the \dProb{} problem (\ref{eq:problem}).
\end{theorem}
\begin{proof}First, any formula
$\varphi_i \in \BExp{H_i}$ can be equivalently written as
$u \wedge \varphi_{i,u} \vee \bar{u} \wedge \varphi_{i,\bar{u}}$ where
$\varphi_{i,u} \isdef \varphi_i[u \mapsto \top]\in \BExp{H_i \setminus\{u\}}$
and
$\varphi_{i,\bar{u}} \isdef \varphi_{i}[u \mapsto \bot] \in \BExp{H_i \setminus \{u\}}$.
Second, we can prove the equivalence:
\begin{eqnarray*}
\Phi[x_i \mapsto \varphi_i] & \Leftrightarrow &
(u \wedge \Phi_u \vee \bar{u} \wedge \Phi_{\bar{u}})[x_i \mapsto u \wedge \varphi_{i,u} \vee \bar{u} \wedge \varphi_{i,\bar{u}}] \\
& \Leftrightarrow & u \wedge \Phi_u[x_i \mapsto \varphi_{i,u}] \vee
 \bar{u} \wedge \Phi_{\bar{u}}[x_i \mapsto \varphi_{i,\bar{u}}]
\end{eqnarray*}
by considering the decomposition of $\Phi_{u}$, $\Phi_{\bar{u}}$ according to the variable $x_i$.
The equivalence above can then be generalized to a complete substitution 
$\sigma_X = \{ x_i \mapsto \varphi_i\}_{i \in [1,n]}$ of
maximizing variables.  Let us denote respectively $\sigma_{X,u} \isdef \{ x_i \mapsto \varphi_{i,u} \}_{i\in[1,n]}$,
$\sigma_{X, \bar{u}} \isdef \{x_i \mapsto \varphi_{i,\bar{u}}\}_{i \in [1,n]}$.  Therefore, one obtains
\begin{eqnarray*}
\Phi[\sigma_X] & \Leftrightarrow & (u \wedge \Phi_u \vee \bar{u} \wedge \Phi_{\bar{u}}) [x_i \mapsto \varphi_i]_{i\in [1,n]}\\
& \Leftrightarrow & u \wedge \Phi_u[x_i \mapsto \varphi_{i,u}]_{i\in[1,n]} \vee
 \bar{u} \wedge \Phi_{\bar{u}}[x_i \mapsto \varphi_{i,\bar{u}}]_{i\in[1,n]} \\
& \Leftrightarrow & u \wedge \Phi_u[\sigma_{X,u}] \vee \bar{u} \wedge \Phi_{\bar{u}}[\sigma_{X,\bar{u}}]
\end{eqnarray*}
Third, the later equivalence provides a way to compute the number of $Y$-models
of the formula $\exists Z.~\Phi[\sigma_Z]$ as follows:
\begin{eqnarray*}
\cardinal{\exists Z.~\Phi[\sigma_X]}{Y}
& = & \cardinal{\exists Z.~(u \wedge \Phi_u[\sigma_{X,u}] \vee
\bar{u} \wedge \Phi_{\bar{u}}[\sigma_{X,\bar{u}}])}{Y} \\
& = & \cardinal{\exists Z.~(u \wedge \Phi_u[\sigma_{X,u}]) \vee
\exists Z.~(\bar{u} \wedge \Phi_{\bar{u}}[\sigma_{X,\bar{u}}])}{Y} \\
& = & \cardinal{\exists Z.~(u \wedge \Phi_u[\sigma_{X,u}])}{Y} +
\cardinal{\exists Z.~(\bar{u} \wedge \Phi_{\bar{u}}[\sigma_{X,\bar{u}}])}{Y} \\
& = & \cardinal{\exists Z\setminus\{u\}.~\Phi_u[\sigma_{X,u}]}{Y\setminus\{u\}} +
\cardinal{\exists Z\setminus\{u\}.~\Phi_{\bar{u}}[\sigma_{X,\bar{u}}]}{Y \setminus\{u\}}
\end{eqnarray*}
Note that the third equality holds only because $u \in Y$ or 
$u \in Z$ and functionally dependent on counting variables $Y$.  Actually,
in these situations, the sets of $Y$-projected models of 
respectively, $u \wedge \Phi_u[\sigma_{X,u}]$ and
$\bar{u} \wedge \Phi_{\bar{u}}[\sigma_{X,\bar{u}}]$ are disjoint.  Finally,
  the last equality provides the justification of the theorem, that is, finding
$\sigma_X$ which maximizes the left hand side reduces to finding
$\sigma_{X,u}$, $\sigma_{X,\bar{u}}$ which maximizes independently the two
terms of right hand side, and these actually are the solutions of the two
$u$-reduced problems (\ref{eq:problem:u}) and
(\ref{eq:problem:baru}). \squareforqed
\end{proof}
%
%
\begin{example}
Let us reconsider \cref{ex:running1}.  It is an
immediate observation that existential variables $z_1$, $z_2$ are
functionally dependent on counting variables $y_1$, $y_2$ according to the
objective.  Therefore the local method is applicable and henceforth since
$H_1 = \{z_1, z_2\}$ one reduces the initial problem to four smaller
problems, one for each valuation of $z_1$, $z_2$, as follows:
\begin{eqnarray*}
z_1 \mapsto \top, z_2 \mapsto \top & :~ & \max\nolimits^{\emptyset} x_1.~ \counting y_1.~\counting y_2~.
(x_1 \Leftrightarrow y_1) \land (\top \Leftrightarrow y_1 \lor y_2) \land (\top \Leftrightarrow y_1 \land y_2) \\
z_1 \mapsto \top, z_2 \mapsto \bot & :~ & \max\nolimits^{\emptyset} x_1.~ \counting y_1.~\counting y_2~.
(x_1 \Leftrightarrow y_1) \land (\top \Leftrightarrow y_1 \lor y_2) \land (\bot \Leftrightarrow y_1 \land y_2) \\
z_1 \mapsto \bot, z_2 \mapsto \top & :~~ & \max\nolimits^{\emptyset} x_1.~ \counting y_1.~\counting y_2~.
(x_1 \Leftrightarrow y_1) \land (\bot \Leftrightarrow y_1 \lor y_2) \land (\top \Leftrightarrow y_1 \land y_2) \\
z_2 \mapsto \bot, z_2 \mapsto \bot & :~~ & \max\nolimits^{\emptyset} x_1.~ \counting y_1.~\counting y_2~.
(x_1 \Leftrightarrow y_1) \land (\bot \Leftrightarrow y_1 \lor y_2) \land (\bot \Leftrightarrow y_1 \land y_2) 
\end{eqnarray*}
The four problems are solved independently and have solutions e.g., respectively $x_1 \mapsto c_1 \in \{ \top \}$, $x_1 \mapsto c_2 \in \{\top, \bot\}$, $x_1 \mapsto c_3\in \{\top, \bot\}$, $x_1 \mapsto c_4 \in \{ \bot\}$.
By recombining these solutions according to \cref{thm:local} one obtains several solutions to the original \dProb{} problem of the form:
\begin{equation}\nonumber
x_1 \mapsto \pars{z_1 \wedge z_2 \wedge c_1} \vee
\pars{z_1 \wedge \bar{z_2} \wedge c_2} \vee
\pars{\bar{z_1} \wedge z_2 \wedge c_3} \vee
\pars{\bar{z_1} \wedge \bar{z_2} \wedge c_4}
\end{equation}
They correspond to solutions
already presented in \cref{ex:global}, that is:
\begin{align*}
x_1& \mapsto
\pars{z_1 \wedge z_2 \wedge \top} \vee
\pars{z_1 \wedge \overline{z_2} \wedge \bot} \vee
\pars{\overline{z_1} \wedge z_2 \wedge \bot} \vee
  \pars{\overline{z_1} \wedge \overline{z_2} \wedge \bot} &(\equiv z_1 \wedge z_2)  \\
x_1& \mapsto
\pars{z_1 \wedge z_2 \wedge \top} \vee
\pars{z_1 \wedge \overline{z_2} \wedge \bot} \vee
\pars{\overline{z_1} \wedge z_2 \wedge \top} \vee
  \pars{\overline{z_1} \wedge \overline{z_2} \wedge \bot} &(\equiv z_2) \\
x_1& \mapsto
\pars{z_1 \wedge z_2 \wedge \top} \vee
\pars{z_1 \wedge \overline{z_2} \wedge \top} \vee
\pars{\overline{z_1} \wedge z_2 \wedge \bot} \vee
  \pars{\overline{z_1} \wedge \overline{z_2} \wedge \bot} &(\equiv z_1) \\
x_1& \mapsto
\pars{z_1 \wedge z_2 \wedge \top} \vee
\pars{z_1 \wedge \overline{z_2} \wedge \top} \vee
\pars{\overline{z_1} \wedge z_2 \wedge \top} \vee
  \pars{\overline{z_1} \wedge \overline{z_2} \wedge \bot} &(\equiv z_1 \vee z_2)
\end{align*}
\end{example}
Finally, note that the local resolution method has potential for
parallelization.  It is possible to eliminate not only one but all
common variables in the dependency sets as long as they fulfill the
required property.  This leads
to several strictly smaller sub-problems that can be solved in parallel.
The situation has been already illustrated in the previous example,
where by the elimination of $z_1$ and $z_2$ one obtains 4 smaller
sub-problems.

\subsection{Solving variables with no dependencies}\label{subsec:local2}
 Let us consider now a
maximizing variable which has an empty dependency set.
Without lack of generality, assume $x_1$ has an empty dependency set, i.e. $H_1 = \emptyset$. Thus, the only possible values that can be assigned to $x_1$ are $\top$ or $\bot$.  
 Let us consider the following
$x_1$-reduced \dProb{} problems:
\begin{eqnarray*}
\max\nolimits^{H_2} x_2.~\dots ~\max\nolimits^{H_n} x_n.~\counting\,Y.~\exists\,Z.~\Phi_{x_1} \\
\max\nolimits^{H_2} x_2.~\dots ~\max\nolimits^{H_n} x_n.~\counting\,Y.~\exists\,Z.~\Phi_{{\overline{x_1}}}
\end{eqnarray*}

and let $\sigma^*_{X,x_1}$, $\sigma^*_{X,{\overline{x_1}}}$ denote respectively the solutions to the problems above.
The following \lcnamecref{proposition:local} is easy to prove, and provides the solution of the original problem
based on the solutions of the two smaller sub-problems.
\begin{proposition}\label{proposition:local}
 The substitution  $\sigma_X^*$ defined as
\[
  \sigma_X^*\isdef \begin{cases}
    \sigma^*_{X,x_1} \uplus \{x_1\mapsto \top\}  & \text{if } \cardinal{\exists Z.~ \Phi_{x_1}[\sigma^*_{X,x_1}]}{Y} \geq \cardinal{\exists Z.~ \Phi_{\overline{x_1}}[\sigma^*_{X,\overline{x_1}}]}{Y}\\
    \sigma^*_{X,\overline{x_1}} \uplus \{x_1\mapsto \bot\}& \text{otherwise }  \\
  \end{cases}
\]
  is a solution to the \dProb{} problem (\ref{eq:problem}).
\end{proposition}

\section{Application to Software Security}\label{sec:appsec}


In this section, we give a concrete application of \dProb{} in the context of \textit{software security}.
More precisely, we show that  finding an optimal strategy for an adaptative attacker trying to break the security of some program 
can be naturally encoded as specific instances of the \dProb{} problem.

In our setting, we allow the attacker to interact multiple times with the target program. Moreover,
we assume that the adversary is able to make \emph{observations}, either from the legal outputs or
using some side-channel leaks.  Adaptive attackers~\cite{Dullien17,saha2021incatksynth,phan2017sidechanatk} are a special form of active
attackers considered in security that are able to select their inputs based on former observations,
such that they maximize their chances to reach their goals (i.e., break some security properties).

First we present in more details this attacker model we consider, and then we focus on two representative attack 
objectives the attacker aims to maximize:
\begin{itemize}
  \item either the probability of reaching a specific point in the target program, while satisfying some objective function (Section \ref{appl1}), 
  \item or the amount of information it can get about some fixed secret used by the program (Section \ref{appl2}). 
\end{itemize}
At the end of the section, we show that the improvements presented in the previous sections apply
in both cases.

\subsection{Our model of security in presence of an adaptive adversary}\label{appl0}

The general setting we consider is the one of so-called \textit{active} attackers, 
able to provide \textit{inputs} to the program they target. Such attacks are then said \textit{adaptive} 
when the attacker is able to deploy an attack strategy, which continuously relies on some knowledge 
gained from previous interactions with the target program, and allowing to maximize its chances of success.
Moreover, we consider the more powerful attacker model where the adversary is assumed to know the code of the target
program.

Note that such an attacker model is involved in most recent concrete attack scenarios, where launching an exploit
or disclosing some sensitive data requires to chain several (interactive) attack steps in order to defeat some protections and/or 
to gain some intermediate privileges on the target platform\comm{\todo{add refs !}}. Obviously, from the defender side, quantitative
measures about the ``controllability'' of such attacks is of paramount importance for exploit analysis or vulnerability triage.

When formalizing the process of \emph{adaptatively attacking} a given program, one splits the
program's variables between those \emph{controlled} and those \emph{uncontrolled} by the attacker.
Among the \emph{uncontrolled} variables one further distinguishes those \emph{observable} and those \emph{non-observable}, 
the former ones being available to the attacker for producing its (next) inputs.
The \emph{objective} of the attacker is a formula, depending on the values of program variables, and determining whether 
the attacker has successfully conducted the attack.

For the sake of simplicity -- in our examples -- we restrict ourselves to non-looping sequential
programs operating on variables with bounded domains (such as finite integers, Boolean's, etc).  
We furthermore consider the programs are written in SSA form, assuming that each variable is assigned before it is used. 
These hypothesis fit well in the context of a code analysis technique like \textit{symbolic execution}~\cite{king1976symbolic}, 
extensively used in software security.

Finally, we also rely on explicit (user-given) annotations by predefined functions (or macros) to identify the different classes of program
variables and the attacker's objective.  
In the following code excerpts, we assume that:
\begin{itemize}
	\item The \texttt{random} function produces an uncontrolled non-observable value;  
it allows for instance to simulate the generation of both long term keys and nonces in a program using cryptographic primitives.
	\item The \texttt{input} function feeds the program with an attacker-controlled value.
	\item The \texttt{output} function simulates an observation made by the adversary and 
		denotes a value obtained through the evaluation of some expression of program variables.
\end{itemize}

\subsection{Security as a rechability property}\label{appl1}
We show in this section how to encode  \emph{quantitative reachability} defined in
\cite{bardin2022quantreach} as an instance of the  \dProb{} problem.

In \emph{quantitative reachability}, the goal of an adversary is to reach some target location in
some program such that some \emph{objective property} get
satisfied.  In order to model this target location of the program that the attacker wants to reach,
we extend our simple programming language with a distinguished \texttt{win} function. The
\texttt{win} function can take a predicate as argument (the objective  property) and is
omitted whenever this predicate is the \texttt{True} predicate. In practice such a predicate 
may encode some extra conditions required to trigger and exploit some vulnerability at the given program 
location (e.g., overflowing a buffer with a given payload).

\begin{example} \label{ex:sec:simpleatk}
In Program~\labelcref{alg:simpleatk} one can see an example of annotated program.  $y_1$
and $y_2$ are uncontrollable non-observable variables.  $z_1$ is an
observable variable holding the sum $y_1 + y_2$.  $x_1$ is a variable
controlled by the attacker.  The \emph{attacker's objective} corresponds to
the \emph{path predicate} \( y_1 \leq x_1 \) denoting
the condition to reach the \texttt{win} function call and the \emph{argument predicate} \( x_1 \leq y_2 \) denoting the \emph{objective property}.  Let us observe that a
successful attack exists, that is, by taking $x_1 \leftarrow \frac{z_1}{2}$ the
objective is always reachable.
\end{example}

\begin{algorithm}[t]
  \SetAlgorithmName{Program}{Program}{List of Programs}
  \DontPrintSemicolon
  \SetKwFunction{Input}{input}
  \SetKwFunction{Output}{output}
  \SetKwFunction{Random}{random}
  \SetKwFunction{Win}{win}
  $y_1 \leftarrow \Random{}$ \;
  $y_2 \leftarrow \Random{}$ \;
  \BlankLine{}
  $z_1 \leftarrow \Output(y_1 + y_2)$ \;
  $x_1 \leftarrow \Input{}$\;
  \BlankLine{}
  \If{$y_1 \leq x_1$}{
    \Win{$x_1 \leq y_2$}\;
  }
  \caption{A first program example\label{alg:simpleatk}}
\end{algorithm}


When formalizing adaptive attackers, the \emph{temporality} of interactions
(that is, the order of inputs and outputs) is important, as the attacker
can only synthesize an input value from the output values that were
observed \emph{before} it is asked to provide that input.  To track the
temporal dependencies in our formalization, for every controlled variable $x_i$
one considers the set $H_i$ of observable variables effectively known at
the time of defining $x_i$, that is, representing the accumulation of
attacker's knowledge throughout the interactions with the program.



We propose hereafter a systematic way to express the problem of synthesis
of an optimal attack (that is, with the highest probability of the objective property to get
satisfied), as
a \dProb{} instance.   Let $Y$ (resp. $Z$) be the set of
uncontrolled variables being assigned to \texttt{random()} which in this section is assumed to uniformly sample values in their domain (resp. other
expressions) in the program.  For a variable $z \in Z$ let moreover $e_z$ be
the unique expression assigned to it in the program, either through an
assignment of the form $z \leftarrow e_z$ or
$z \leftarrow \texttt{output}(e_z)$.
 Let $X = \{ x_1, ..., x_n\}$ be the set of controlled
variables with their temporal dependencies respectively subsets $H_1, \ldots , H_n\subseteq Z$ of uncontrollable variables.
Finally, let $\Psi$ be the attacker
objective, that is, the conjunction of the argument of the \texttt{win} function and the path predicate leading to the \texttt{win} function
call. Consider the next most likely generalized \dProb{} problem:
\begin{equation}
\max\nolimits^{H_1}x_1.~...~\max\nolimits^{H_n}x_n.~ \counting Y.~ \exists Z.~ \Psi \land \bigwedge\nolimits_{z \in Z} (z = e_z)
\end{equation}

\begin{example}\label{ex:sec:sum}
Consider the annotated problem from Program~\ref{alg:simpleatk}.
The encoding of the optimal attack leads to the generalized \dProb{} problem:
\begin{eqnarray*}
 \max\nolimits^{\{z_1\}}x_1.~ 
  \counting y_1.~ \counting y_2.~\exists z_1.~ 
 (y_1 \leq x_1 \land x_1 \leq y_2) \land (z_1 = y_1 + y_2 )
\end{eqnarray*}
\end{example}

Note that in contrast to the \dProb{} problem (\ref{eq:problem}), the variables are not restricted
to Booleans (but to some finite domains) and the expressions are not restricted to Boolean terms
(but involve additional operators available in the specific domain theories e.g., $=$, $\geq$, $+$,
$-$, etc). Nevertheless, as long as both variables and additional operators can be respectively,
represented by and interpreted as operations on bitvectors, one can use \emph{bitblasting} and
transforms the generalized problem into a full-fledged \dProb{} problem and then solve it by the
techniques introduced earlier in the paper.  

Finally, note also that in the \dProb{} problems constructed as above, the
maximizing variables are dependent by definition on existential variables only.
Therefore, as earlier discussed in \cref{sec:problem}, these problems
cannot be actually reduced to similar \dPds{} problems.  However, they
compactly encode the quantitative reachability properties subject to input/output dependencies.

\subsection{Security as a lack of leakage property}\label{appl2}

In this section, we extend earlier work on adaptive attackers from \cite{saha2021incatksynth} by
effectively synthesizing the \emph{strategy} the attacker needs to deploy in order to maximize its
knowledge about some secret value used by the program. Moreover, we show that in our case, we are
able to keep symbolic the trace corresponding to the attack strategy, while in
\cite{phan2017sidechanatk}, the attacker strategy is a concretized tree, which explicitly states,
for each concrete program output, what should be the next input provided by the adversary.
Following ideas proposed in \cite{phan2017sidechanatk}, symbolic execution can be used to generate 
constraints characterizing partitions on the secrets values, where each partition corresponds to the set of secrets leading to
the same \emph{sequences} of side-channel observations.

\noindent \begin{minipage}[t]{.5\textwidth}
\vspace{0pt}
\begin{algorithm}[H]
  \SetAlgorithmName{Program}{Program}{List of Programs}
  \DontPrintSemicolon
  \SetKwFunction{Input}{input}
  \SetKwFunction{Output}{output}
  \SetKwFunction{Random}{random}
  \SetKwFunction{Win}{win}
  $z \leftarrow \Random{}$  //the secret \;
  $x \leftarrow \Input{}$\;
  \BlankLine{}
  \eIf{$x \geq z$}{
    ... some computation taking 10 seconds \;
  }{
    ... some computation taking 20 seconds \;
  }
  \caption{A leaking program\label{twoL}}
\end{algorithm}
\end{minipage}
\hfill{}
\begin{minipage}[t]{.45\textwidth}
\vspace{0pt}
\begin{algorithm}[H]
  \SetAlgorithmName{Program}{Program}{List of Programs}
  \SetKwFunction{Input}{input}
  \SetKwFunction{Output}{output}
  \SetKwFunction{Random}{random}
  \SetKwFunction{Win}{win}
  $z \leftarrow \Random{}$ \;
  \BlankLine{}
  $x_1 \leftarrow \Input{}$\;
  $y_1 \leftarrow \Output(x_1 \geq z)$\;
  \BlankLine{}
  $x_2 \leftarrow \Input{}$\;
  $y_2 \leftarrow \Output(x_2 \geq z)$\;
  \BlankLine{}
  $x_3 \leftarrow \Input{}$\;
  $y_3 \leftarrow \Output(x_3 \geq z)$\;
  \caption{An iterated leaking program\label{out6}}
\end{algorithm}
\end{minipage}

\newcommand{\Concat}{\mathbin{{+}{+}}}

\begin{example} \label{ex:capacity6}
  Let us consider the excerpt Program \ref{twoL} taken from \cite{phan2017sidechanatk}. This program
	is not \textit{constant-time}, namely it executes a branching instruction whose condition depends on the secret $z$. Hence an adversary
  able to learn the branch taken during  the execution, either by measuring the time or doing some
  cache-based attack, will get some information about the secret $z$. A goal of an adversary
  interacting several times with the program could be to maximize the amount of information leaked about
  the secret value $z$. When the program is seen as a channel leaking information, the channel
  capacity theorem \cite{smith2009foundations} states that the information leaked by a program is
  upper-bounded by the number of different observable outputs of the program (and the maximum is
  achieved whenever the secret is the unique randomness used by the program). In our case, it means
  that an optimal adaptive adversary interacting $k$-times with the program should maximize the
  number of   different observable outputs. Hence, for example, if as in \cite{phan2017sidechanatk},
  we fix $k=3$ and if we assume that the secret $z$ is uniformly sampled in the domain $1\leq z \leq
  6$, then the optimal strategy corresponds to maximize the number of different observable outputs
  $y$ of the Program \ref{out6}, which corresponds to the following  \dProb{} instance: 
  \begin{multline*}
    \Max{\emptyset}{x_1} \Max{\{y_1\}}{x_2} \Max{\{y_1, y_2\}}{x_3} \counting y_1 .~ \counting y_2.~
    \counting y_3.~ \exists z ~. \\
    (y_1 \Leftrightarrow x_1 \geq z) \land (y_2 \Leftrightarrow x_2 \geq z) \land (y_3 \Leftrightarrow x_3 \geq z) \land (1 \leq z \leq 6)
  \end{multline*}
  Our prototype provided the following solution:
   $x_1 = 100,  $
    $x_2 =  y_1 1 0,  $
    $x_3 =  y_1 y_2 1,  $
  that basically says: the attacker should first input 4 to the program, then the input corresponding
  to the integer whose binary encoding is  $y_1$ concatenated with $10$, and the last input $x_3$ is
  the input corresponding to the integer whose binary encoding is the concatenation of $y_2$, $y_1$
  and $1$. In \cite{phan2017sidechanatk} the authors obtain an equivalent
  attack encoded as a tree-like strategy of concrete values. 
\end{example}

We now show a systematic way to express the problem of the synthesis
of an optimal attack expressed as the maximal channel capacity of a program seen as an information leakage channel, as a \dProb{} instance.   Contrary to the previous section, the roles of $Y$ and $Z$ are now switched:  $Y$ is a set of variables encoding the observables output by the program;   $Z$  is  the set of  variables uniformly sampled by \texttt{random()} or assigned to  other expressions in the program.  For a variable $y \in Y$, let  $e_y$ be
the unique expression assigned to it in the program through an
assignment of the form $y \leftarrow \texttt{output}(e_y)$. For a variable $z \in Z$, let moreover $e_z$ be the unique expression assigned to it in the program  through an
assignment of the form $z \leftarrow e_z$ or the constraint encoding the domain used to sample values in 
$z \leftarrow \texttt{random}()$.
 Let $X = \{ x_1, ..., x_n\}$ be the set of controlled
variables with their temporal dependencies respectively subsets $H_1, \ldots , H_n\subseteq Y$. Consider now the following most likely generalized \dProb{} problem:

\begin{equation*}
\max\nolimits^{H_1}x_1.~...~\max\nolimits^{H_n}x_n.~ \counting Y.~ \exists Z.~  \bigwedge\nolimits_{y \in Y} (y = e_y) \land \bigwedge\nolimits_{z \in Z} (z = e_z)
\end{equation*}

\comm{
\begin{example}[Maximizing channel capacity] \label{ex:capacity}
  Maximizing the channel capacity in example \missingref{} corresponds to the following \dProb{} instance:
  \begin{multline*}
    \Max{\emptyset}{x_1} \Max{\{y_1\}}{x_2} \Max{\{y_1, y_2\}}{x_3} \counting y_1 .~ \counting y_2.~
    \counting y_3.~ \exists z ~. \\
    (y_1 \Leftrightarrow x_1 \geq z) \land (y_1 \Leftrightarrow x_2 \geq z) \land (y_1 \Leftrightarrow x_3 \geq z)
  \end{multline*}

  For which the answer is:
  \begin{eqnarray*}
    x_1 &=& 100 \\
    x_2 &=& y_1 \Concat{} 10 \\
    x_3 &=& y_1 \Concat{} y_2 \Concat{} 1 \\
  \end{eqnarray*}
\end{example}
}

Finally, in contrast to reachability properties, in the \dProb{} problems
obtained as above for evaluating leakage properties, the maximizing
variables are by definition dependent on counting variables only.  Consequently, for
these problems, the existential variables can be apriori eliminated so that
to obtain an equivalent \dPds{}\footnotemark{} problem as discussed in \cref{sec:problem}.
\footnotetext{
  Actually these problems can even be reduced to \dPssat{} instances.
}

\subsection{Some remarks about the applications to security}\label{appl3}

\comm{
\retq{do we need this ? : After bit-blasting, the patterns in the resulting formula that correspond to specific operations in
the original formula (i.e. multiplications, additions) are not easily recoverable. Similarly,
patterns in the answer to the \dProb{} on the bit-blasted formula is not easily converted back to a
given set of operations in the original formula's language.}
}

Let us notice some interesting properties of the attacker
synthesis's \dProb{} problems.  If controlled variables $x_1$, $x_2$, ...,
$x_n$ are input in this order within the program then necessarily
$H_1 \subseteq H_2 \subseteq ... \subseteq H_n$.  That is, the knowledge of
the attacker only increases as long as newer observable values became
available to it.
Moreover, since we assumed that variables are used only after they were initialized, the sets $H_i$
contain observable variables that are dependent only on the counting variables $Y$.
Hence we can apply iteratively the following steps from the local resolution method described
in \cref{sec:local}:
\begin{itemize}
\item While $H_1 \not=\emptyset$, \comm{and  (that
is, now a static condition easy to check on the program)}  apply
 the local resolution method described
in \cref{subsec:local1} iteratively until $H_1$ becomes empty. For example, it is the case of \cref{ex:sec:simpleatk} where $z_1$ is dependent only on counting
variables $y_1$ and $y_2$.

\item When $H_1$ becomes $\emptyset$,   apply the  local resolution method described in \cref{subsec:local2} in order to eliminate the first maximizing variable. 
\end{itemize}

\section{Implementation and Experiments}\label{sec:bench}

\newcommand{\BaxMC}{\textsc{BaxMC}}
\newcommand{\Dmax}{\textsc{D4max}}
\newcommand{\DSSATpre}{\textsc{DSSATpre}}

We implement \cref{alg:incremental} leaving generic the choice of the underlying \dPmd{} solver.
For concrete experiments, we used both the approximate solver \BaxMC\footnotemark~\cite{vigouroux2022baxmc} and the
exact solver \Dmax{}~\cite{audemard2022d4max}.
\footnotetext{
Thanks to specific parametrization and the
oracles~\cite{chakraborty2013approxmc} used internally by \BaxMC{}, it can be considered an exact solver on the
small instances of interest in this section.
}

In the implementation of \cref{alg:incremental} in our tool, the filter $\Psi'$ is handled as discussed at
the end of \cref{sec:incremental}: the formula effectively solved is $\Phi' \land \neg \Psi'$,
allowing to use any \dPmd{} solver without any prior modification.
Remark that none of \BaxMC{} and \Dmax{} originally supported exploiting the $\alpha_0$ parameter of
\cref{alg:incremental} out of the box. While \Dmax{} is used of the shelf, we modified \BaxMC{} to actually support
this parameter for the purpose of the experiment.

We use the various examples used in this paper as benchmark instances for the implemented tool.
\cref{ex:running1,ex:running2} are used as they are. We furthermore use \cref{ex:running_old} (in
appendix) which is a slightly modified version of \cref{ex:running1}. We
consider \cref{ex:sec:sum,ex:capacity6} from \cref{sec:appsec} and perform the following steps to
convert them into \dProb{} instances:
\begin{inparaenum}[(i)]
  \item bitblast the formula representing the security problem into a \dProb{} instance over boolean
        variables;
  \item solve the later formula;
  \item propagate the synthesized function back into a function over bit-vectors for easier visual
        inspection of the result.
\end{inparaenum}

We also add the following security related problems (which respectively correspond to
Program~\labelcref{pgm:alternating} in appendix and a relaxed version of \cref{ex:capacity6} in
\cref{sec:appsec}) into our benchmark set:
\begin{example} \label{bench:guessbits}
  \( \Max{\emptyset}{x_1} \Max{\{z_1\}}{x_2} \Max{\{z_1, z_2\}}{x_3}
    \counting y_1.~ \exists z_1.~ \exists z_2.~ \) \\
    \hspace*{4cm} \( (x_3 = y_1) \land (z_1 = x_1 \geq y_1 \land z_2 = x_2 \geq y) \)
\end{example}
\begin{example} \label{bench:capacity}
    \( \Max{\emptyset}{x_1} \Max{\{y_1\}}{x_2} \Max{\{y_1, y_2\}}{x_3} \counting y_1 .~ \counting y_2.~
    \counting y_3.~ \exists z ~. \) \\
    \hspace*{4cm} \( (y_1 \Leftrightarrow x_1 \geq z) \land (y_2 \Leftrightarrow x_2 \geq z) \land (y_3 \Leftrightarrow x_3 \geq z) \)
\end{example}

When bitblasting is needed for a given benchmark, the number of bits used for bitblasting is
indicated in parentheses. After the bitblasting operation, the problems can be considered medium
sized.
\begin{table}[t]
  \centering
  \caption{Summary of the performances of the tool. $\size{\Phi}$ denotes the number of clauses. The
  last two columns indicate the running time using the specific \dPmd{} oracle.}
  \begin{tabular}{|l|r|r|r|r|r|r|}
    \hline{}
    Benchmark name & $\size{X}$ & $\size{Y}$ & $\size{Z}$ & $\size{\Phi}$ & Time (\BaxMC{}) & Time (\Dmax{}) \\ \hline \hline
    \cref{ex:running1} & 1 & 2 & 2 & 7 & 32ms & 121ms \\ \hline
    \cref{ex:running2} & 2 & 2 & 2 & 7 & 25ms & 134ms \\ \hline
    \cref{ex:running_old} & 1 & 2 & 1 & 5 & 16ms & 89ms \\ \hline \hline
    \cref{ex:sec:sum} (3 bits) & 3 & 6 & 97 & 329 & 378ms & 79.88s \\ \hline
    \cref{ex:sec:sum} (4 bits) & 4 & 8 & 108 & 385 & 638.63s & > 30mins \\ \hline
    \cref{ex:capacity6} (3 bits) & 9 & 3 & 150 & 487 & 18.78s & 74.58s \\ \hline
    \cref{bench:guessbits} (3 bits) & 9 & 3 & 93 & 289 & 74.00s & 18.62s \\ \hline
    \cref{bench:capacity} (3 bits) & 9 & 3 & 114 & 355 & 9.16s & 93.48s \\ \hline
  \end{tabular}
  \label{tab:perf}
\end{table}

As you can see in \cref{tab:perf}, the implemented tool can effectively solve all the examples
presented in this paper. The synthesized answers (i.e. the monomials selected in \cref{alg:incremental},
\cref{alg:incremental:final}) returned by both
oracles are the same.

For security examples, one key part of the process is the translation of the synthesized answer
(over boolean variables) back to the original problem (over bit-vectors). In order to do that, one
can simply concatenate the generated sub-functions for each bit of the bit-vector into a complete
formula, but that would lack explainability because the thus-generated function would be a
concatenation of potentially big sums of monomials. In order to ease visual inspection, we run a
generic simplification step~\cite{gario2015pysmt} for all the synthesized sub-function, before
concatenation. This simplification allows us to directly derive the answers explicited in
\cref{ex:sec:sum,ex:capacity6} instead of their equivalent formulated as sums of monomials, and
better explain the results returned by the tool.

Unfortunately, we could not compare our algorithm against the state-of-the-art \dPds{} solver
\DSSATpre{}~\cite{lee2021dssat} on the set of example described in this paper because
\begin{inparaenum}[(i)]
  \item as discussed in \cref{sec:hard:dssat}, some \dProb{} instances cannot be converted into
    \dPds{} instances,
  \item for the only \dProb{} instance (\cref{bench:capacity}) that can be converted into a \dPds{}
    instance, we were not able to get an answer using \DSSATpre{}.
\end{inparaenum}

\section{Related Work}\label{sec:relwork}
As shown in \cref{sec:problem}, \dProb{} subsumes the \dPds{} and \dPdq{} problems. This relation
indicates a similarity of the three problems, and 
thus some related works can be extracted from here. From the complexity point of view,
the decision version of \dProb{} can be shown to be $\mathsf{NEXPTIME}$-complete and hence it
lies in the same complexity class as  \dPdq{}~\cite{peterson2001dqbfcompl} and
\dPds{}~\cite{lee2021dssat}.

Comparing the performances of existing \dPdq{} algorithms with the proposed
algorithms for \dProb{} is not yet realistic since they address different
objectives.  However, one can search for potential improvements for
solving \dProb{} by considering the existing enhancements proposed
in~\cite{kovasznai2016statedqbf} to improve the resolution of \dPdq{}.
For example, dependency schemes~\cite{wimmer2016dqbfdeps} are a way to change
the dependency sets in \dPdq{} without changing the \emph{truth value}
compared to the original formula. Thus, adaptations of these dependency
schemes could be applied to our problem as well and potentially lead to a
significant decrease of the size of the resulting \dPmd{}
problems.


The \dPds{} problem is currently receiving an increased attention by the
research community.  A first sound and complete resolution procedure has
been proposed in \cite{luo2023dssatres}, however, without being yet
implemented.  The only available \dPds{} solver nowadays is 
\DSSATpre{}~\cite{cheng2023dssatpre}.  This tool relies on preprocessing to
get rid of dependencies and to produce equivalent \dPssat{} problems.  These
problems are then accurately solved by existing \dPssat{} solvers
\cite{chen2021claussat,wang2022elimssat}, some
of them being also able to compute the optimal assignments for maximizing
variables.  In contrast, our tool for solving \dProb{} relies on existing
\dPmd{} solvers, always synthesizes the assignments for maximizing
variables and provide support for approximate solving.  Moreover, due to
the presence of existential variables, note that \dProb{} and \dPds{} are
fundamentally different problems.  Existential variables are already pinpointing the
difference between the two pure counting problems \Problem{\#SAT}
and \Problem{\#$\exists$SAT}~\cite{aziz2015projmc,lagniez2019projmc}.  In cases where
maximizing variables depend on existential variables no trivial reduction
from \dProb{} to \dPds{} seems to exists.


From the security point of view, the closest works to our proposal are the ones decribed
in~\cite{phan2017sidechanatk,saha2021incatksynth}. As the authors in these papers, we are able  to
effectively synthesize the optimal adaptive strategy the attacker needs to deploy in order to
maximize its knowledge about some secret value used by the program. In addition,  we show that in
our case, we are able to keep symbolic the trace corresponding to the attack strategy, while in
\cite{phan2017sidechanatk}, the attacker strategy is a concretized tree which explicitely states,
for each concrete program output, what should be the next input provided by the adversary.

\section{Conclusions}\label{sec:conclusions}
We exposed in this paper a new problem called \dProb{} that subsumes both \dPdq {} and \dPds{}. We
then devised three different resolution methods based on reductions to \dPmd{} and showed the
effectiveness of one of them, the incremental method, by implementing a prototype solver. 
A concrete application of \dProb{} lies in the context of software security, in order to assess the robustness 
of a program by synthesizing the optimal adversarial strategy of an adaptive attacker.

Our work can be expanded in several directions. First, we would like to enhance our prototype with strategies for dependency
expansion in the incremental algorithm. Second, we plan to integrate the local resolution
method in our prototype. Third, we shall apply these techniques
on more realistic security related examples, and possibly getting 
further improvement directions from this dedicated context.

\bibliographystyle{splncs04}
\bibliography{./biblio.bib}

\appendix{}
\section*{Appendix}

\begin{example}\label{ex:running_old}
  Consider the problem:
  \begin{equation}\nonumber
  \max\nolimits^{\{z_1\}}x_1.~ \counting y_1.~\counting y_2.~ \exists z_1.~ (x_1 \Leftrightarrow y_1) \land (z_1 \Leftrightarrow (y_1 \lor y_2))
  \end{equation}
  Let $\Phi_1$ denote the objective formula.  As $\BExp{\{z_1\}} = \{\top,
  \bot, z_1, \overline{z_1} \}$ one shall consider these four possible
  substitutions for the maximizing variable $x_1$ and compute the
  associated number of $\{y_1,y_2\}$-projected models.  For instance,
  $\Phi_1[x_1 \mapsto \bot] \equiv \overline{y_1} \land (z_1
  \Leftrightarrow y_2)$ has two models, respectively $\{y_1\mapsto \bot,
  y_2\mapsto \top, z_1 \mapsto \top\}$ and $\{y_1 \mapsto\bot, y_2\mapsto
  \bot, z_1\mapsto \bot\}$ and two $\{y_1,y_2\}$-projected models
  respectively $\{y_1\mapsto \bot, y_2\mapsto \top\}$ and $\{y_1
  \mapsto\bot, y_2\mapsto \bot\}$.  Therefore $\cardinal{\exists
    z_1.~\Phi_1[x_1 \mapsto \bot]}{\{y_1, y_2\}} = 2$.  The maximizing
  substitution is $x_1 \mapsto z_1$ which has three $\{y_1,
  y_2\}$-projected models, that is $\cardinal{\exists z_1.~\Phi_1[x_1
      \mapsto z_1]}{\{y_1, y_2\}} = 3$.  Note that no substitution for
  $x_1$ exists such that the objective to have four $\{y_1,y_2\}$-projected
  models, that is, always valid for counting variables.
\end{example}

\begin{example}
In Program~\labelcref{pgm:alternating}, one shall know that the optimal strategy is the dichotomic search of $y_1$
  within its possible values.
\begin{algorithm}
  \SetAlgorithmName{Program}{Program}{List of Programs}
  \SetKwFunction{Input}{input}
  \SetKwFunction{Output}{output}
  \SetKwFunction{Random}{random}
  \SetKwFunction{Win}{win}
  $y_1 \leftarrow \Random{}$ \;
  \BlankLine{}
  $x_1 \leftarrow \Input{}$\;
  $z_1 \leftarrow \Output(x_1 \geq y_1)$\;
  \BlankLine{}
  $x_2 \leftarrow \Input{}$\;
  $z_2 \leftarrow \Output(x_2 \geq y_1)$\;
  \BlankLine{}
  $x_3 \leftarrow \Input{}$\;
  \Win{$x_3 \approx^{msb}_{3} y_1$}\;
  \caption{A second program example\label{pgm:alternating}}
\end{algorithm}
\end{example}

\end{document}